\documentclass[11pt]{article} %format=sigconf format=manuscript

\usepackage[margin=1in]{geometry}
\usepackage{amsthm,amsmath,amssymb,amsfonts}
\usepackage[final]{graphicx}
\usepackage{subcaption}
\usepackage{hyperref}
\usepackage{booktabs}
\usepackage{ifdraft}
\usepackage{algorithm, algpseudocode}
\usepackage{xcolor}
\usepackage[utf8]{inputenc}
\usepackage{xspace}
\usepackage{braket}
\usepackage{soul}
\usepackage[capitalize]{cleveref}
\usepackage{balance}
\usepackage{thmtools}
\usepackage{thm-restate}

\DeclareMathOperator*{\argmax}{arg\,max}
\DeclareMathOperator\supp{supp}
\DeclareMathOperator\defeq{\stackrel{\text{def}}{=}}

\newtheorem{theorem}{Theorem}
\newtheorem{claim}{Claim}

\newcommand{\eps}{\varepsilon}
\newcommand{\ulp}{\mathrm{ulp}}

\graphicspath{{figures/}}

\newcommand{\sfunction}[1]{\textsf{\textsc{#1}}}

\title{Precision-based attacks and interval refining: how to break, then fix, differential privacy on finite computers}
\author{
  Samuel Haney\\
  Tumult Labs\\
  \texttt{sam.haney@tmlt.io}
  \and
  Damien Desfontaines\\
  Tumult Labs\\
  \texttt{damien@desfontain.es}
  \and
  Luke Hartman\\
  Tumult Labs\\
  \texttt{luke.hartman@tmlt.io}
  \and
  Ruchit Shrestha\\
  Tumult Labs\\
  \texttt{ruchit.shrestha@tmlt.io}
  \and
  Michael Hay\\
  Tumult Labs\\
  Colgate University\\
  \texttt{michael@tmlt.io}
}

\begin{document}

\maketitle

\begin{abstract}
  Despite being raised as a problem over ten years ago, the imprecision of floating point arithmetic continues to cause privacy failures in the implementations of differentially private noise mechanisms.
  In this paper, we highlight a new class of vulnerabilities, which we call \emph{precision-based attacks}, and which affect several open source libraries.
  To address this vulnerability and implement differentially private mechanisms on floating-point space in a safe way, we propose a novel technique, called \emph{interval refining}.
  This technique has minimal error, provable privacy, and broad applicability.
  We use interval refining to design and implement a variant of the Laplace mechanism that is equivalent to sampling from the Laplace distribution and rounding to a float.
  We report on the performance of this approach, and discuss how interval refining can be used to implement other mechanisms safely, including the Gaussian mechanism and the exponential mechanism.
\end{abstract}

\section{Introduction}

There are many issues that can arise when translating abstract differentially private algorithms to real-world implementations.
One such issue, as observed by Mironov \cite{Mironov.SignificanceLeastSignificant.2012} over ten years ago, is the use of floating point arithmetic.
Mironov demonstrated this issue using the Laplace distribution, commonly used as a building block for DP mechanisms:
he showed that sampling Laplace noise in a naive way creates ``holes'' in the output distribution, and that these holes can be leveraged to leak private data and nullify the privacy guarantee.
Later, similar issues were found in other common primitives, like the exponential mechanism~\cite{ilvento2020implementing} and the Gaussian mechanism~\cite{jin2021we}.
The researchers exposing these vulnerabilities also proposed mitigations, but privacy libraries tend not to use these mitigations: some are complex, not generalizable, and in some cases come with a large utility cost.
Instead, the creators of libraries have invented other techniques to mitigate floating point attacks.

\begin{itemize}
  \item Google's differential privacy team developed alternative procedures to add Laplace and Gaussian noise by sampling \emph{discrete} distributions, and rounding to a fixed granularity~\cite{GoogleSecureNoise}. These are implemented in Google's differential privacy libraries~\cite{GoogleDP}.
  \item Holohan et al.~\cite{HolohanBraghin.SecureRandomSampling.2021} proposed sampling Laplace and Gaussian noise in a way that makes it computationally harder for an attacker to reverse-engineer the output, without any rounding. This approach is implemented in diffprivlib~\cite{diffprivlib}, IBM's differential privacy library.
  \item The authors of SmartNoise Core~\cite{SmartNoiseCore} and OpenDP~\cite{OpenDP} implemented yet another approach, which consists of using MPFR~\cite{MPFR}, an arbitrary precision arithmetic library, to generate a \emph{hole-free} unit Laplace or Gaussian distribution~\cite{OpenDPsamplers}. 
\end{itemize}
Unfortunately, we show that the approaches taken by diffprivlib, SmartNoise Core and OpenDP are susceptible to a new class of floating point vulnerabilities, which we call \emph{precision-based attacks}, and describe in \cref{sec:existing-vulnerabilities}.

Google's approach has other drawbacks: (1) the technique cannot be easily generalized to a safe implementation of other mechanisms, like the exponential mechanism, (2) the rounding precision has to be set in advance, which can be a usability issue and limits the mechanism safety to a strict subset of the floating-point range, and (3) it involves a small $\delta$ in the privacy guarantee, which makes it difficult to use with privacy accounting based on pure or zero-concentrated DP.
% here is the issue: https://github.com/opendp/smartnoise-core/issues/242
% there is also this paper: https://www.iacr.org/archive/eurocrypt2006/40040493/40040493.pdf which Frank mentions in his bug report on smartnoise

In this paper, we propose a new technique for safely implementing DP mechanisms in floating-point space.
Our technique, which we call \emph{interval refining}, has three key properties.
\begin{description}
  \item[Minimal error.] It has the same error as if we were sampling the desired distribution exactly, and rounding to the nearest float. The only error comes from this last rounding step, which is unavoidable for any mechanism returning a float.
  \item[Provable privacy.] It provides the same privacy guarantees as the abstract algorithm it implements. As such, it can be used in contexts where the privacy accounting is based on pure DP, zero-concentrated DP, or other definitions forbidding infinite privacy loss. 
  \item[Broad applicability.] It can be used to implement any mechanism in floating-point space. In this paper, we describe our implementation of the Laplace mechanism, but we also outline how the technique can be applied to other mechanisms, such as the Gaussian and exponential mechanisms.
\end{description}

Our technique is conceptually simple, and follows the same idea as inverse transform sampling, where a sample from a uniform distribution is transformed into a sample of an arbitrary probability distribution by passing it as input to the inverse cumulative distribution function (CDF) of this distribution.
But instead of simply generating a sample (which may be an irrational number), we sample an \emph{interval}, initially large and then iteratively refined, until the entire interval would be rounded to the same floating point number, and we can return this float.
Using this technique, we implement the Laplace mechanism (\cref{alg:lap}) in a way that is identical to sampling from the real-valued Laplace distribution and rounding the sample to the next highest float.
This process is illustrated in~\cref{fig:laplace}.

The rest of the paper is structured as follows: In \cref{sec:existing-vulnerabilities}, we describe novel floating point attacks with applications to existing differential privacy libraries.
In \cref{sec:overview-of-technique}, we outline the key ideas underlying our technique, and explain how to simulate sampling from a real-valued distribution and rounding the result to floating-point.
In \cref{sec:safe-laplace-mechanism}, we instantiate our technique and give a precise algorithm implementing sampling from the Laplace distribution, and rounding to the next highest float; we describe our implementation in the Tumult platform~\cite{TumultPlatform} and report on empirical results.  In~\cref{sec:extensions}, we discuss extensions, limitations, and future work.

\section{Precision-based attacks}
\label{sec:existing-vulnerabilities}

In this section, we outline the vulnerabilities we found in existing libraries.
These vulnerabilities, which we call \emph{precision-based attacks}, rely on two simple observations.
First, if a double-precision floating-point number $x$ is such that $2^k \le |x| < 2^{k+1}$ for some integer $k$, then $x$ is a multiple of $2^{k-52}$.
This quantity $2^{k-52}$ is known as the \emph{unit in the last place}~\cite{goldberg1991what} (or \emph{ulp}) of $x$, we denote it by $\ulp_x$. Second, when adding any double to a fixed double $x$, then the output is always a multiple of $\ulp_x/2$.
The proof of these facts can be found in \cref{sec:attack-details}.

This creates a simple vulnerability affecting implementations that add noise to floating-point numbers without any rounding: when adding noise to e.g. $1$, all possible outputs are multiples of $2^{-53}$.
But when adding noise to 0, the output is exactly the value of the sampled noise. 
If it is possible to sample noise that is \emph{not} a multiple of $2^{-53}$, then this creates a distinguishing event between inputs $0$ and $1$.
This is exactly what happens in additive noise mechanisms in diffprivlib, SmartNoise Core, and OpenDP: if the noise value $r$ is such that $|r|<0.5=2^{-1}$, then it might not be a multiple of $2^{-53}$.
If we return it as is (after adding it to 0), the attacker can deduce that the true value was not $1$.
This can happen arbitrarily often as the noise scale gets smaller; with Laplace noise of scale 1 (corresponding to a counting query with $\eps=1$), approximately 25\% of outputs are distinguishing events in diffprivlib, SmartNoise Core, and OpenDP. 

Perhaps surprisingly, precision-based attacks can create vulnerabilities in other algorithms, besides simple additive noise mechanisms.
Consider the mechanism to compute quantiles based on the exponential mechanism, introduced in~\cite{smith2011privacy}.
This algorithm works in three steps: it splits the output space in intervals based on the input data, privately chooses one interval using the exponential mechanism, and returns a uniform number from this interval.
To implement the last step and generate a uniformly random number in an interval $[x,y)$, a naive approach consists of generating a uniformly random number $r$ in $[0,1)$, and returning $x+(y-x) \cdot r$.
The possible precision of the output depends on the value of $x$, which creates the opportunity for precision-based attacks. 
Finding a pair of input databases demonstrating such distinguishing events is a little more involved than for additive noise mechanisms.
In \cref{sec:attack-details}, we provide more detail on this class of vulnerabilities, and show pairs of datasets which lead to distinguishing events for quantile mechanisms in both diffprivlib and SmartNoise Core.

This novel class of attacks shows that mitigating floating-point vulnerabilities is more difficult than it seems.
It suggests that approaches implemented in production-grade software should be formally documented, and provide a proof that they satisfy the desired privacy guarantees.

These vulnerabilities were communicated to, and acknowledged by, the authors of diffprivlib, SmartNoise Core, and OpenDP via personal correspondence and a public bug report~\cite{OpenDPBugReport} in November and December 2021.

\section{Overview of interval refining}
\label{sec:overview-of-technique}

\def\R{\mathbb{R}}
\def\float{\mathsf{float}}

In this section, we give an overview of our interval refining technique.
Our goal is to simulate the following process: sample $X \sim \mathcal{D}$, where $\mathcal{D}$ is a distribution over sample space $\Omega=\R$, and then round the value to a 64 bit floating point number.
Any rounding scheme is acceptable; for the purposes of exposition, we round up. 
If we call $\float : \Omega \rightarrow S$ a function that maps any real number to the next largest 64-bit floating point number, our goal is to sample $X \sim \mathcal{D}$ and output $\float(X)$.

The main complication of this approach is that prior to rounding, the sampled value may be an irrational number and thus cannot be represented with finite memory.
To solve this issue, we need three key ideas.

The first key idea starts with the observation that the float function partitions the sample space into \emph{intervals}, where all elements in each interval are assigned to the same float value:
$\float$ maps elements of the (infinite) sample space to a finite space $S$ (the set of 64-bit floats).
Therefore, our goal is to sample an element $s \in S$ with probability equal to the probability of sampling an element from $\mathcal{D}$ that maps to $s$.
That is, $\Pr[\text{sampling } s] = \Pr_{X \sim \mathcal{D}}[X \in \float^{-1}(s)]$.
This means that instead of sampling $X$, we can sample a progressively finer interval around $X$:
we start from a large interval, and iteratively refine it until all values within the interval map to the same float.
Conceptually, we never sample $X$ directly; instead, we sample enough information about $X$ to determine to which float it should be mapped.

This gives us the following high-level process.
\begin{enumerate}
    \item Set the current interval $I$ to be entire sample space: $I = \Omega = [-\infty, \infty]$
    \item \label{step:split_i} Partition the current interval $I$ into disjoint intervals $I_1, I_2, \dots$ such that $\bigcup_i I_i = I$, and sample interval $I_i$ with probability proportional to $Pr_{X \sim \mathcal{D}}[X \in I_i | X \in I]$. 
    \item Set the current interval $I$ to the selected interval $I_i$.
    \item If all values in $I$ round to the same float $s$ (i.e., $\exists s: I \subseteq \float^{-1}(s)$), then return $s$. \label{step:termination}
    \item Otherwise, repeat from Step \ref{step:split_i}.
\end{enumerate}
For step~\ref{step:split_i}, there are a number of reasonable ways to partition the space.
For reasons that will be explained, we choose to partition $I$ into two equi-probable sub-intervals $I_1$, $I_2$ such that $Pr_{X \sim \mathcal{D}}[X \in I_1] = Pr_{X \sim \mathcal{D}}[X \in I_2]$.
One can show that the process outlined above samples $s$ with the correct probability.

One challenge remains: the interval boundaries might be irrational numbers, which cannot be exactly represented on a computer with finite memory.

This brings us to our second key idea: using \emph{inverse transform sampling}~\cite{InverseTransformSampling}.
The intuition is simple: sampling $U \sim \mathrm{Uniform}(0,1)$ and computing $X = F^{-1}_\mathcal{D}(U)$, where $F^{-1}_\mathcal{D}$ is the inverse cumulative distribution function for distribution $\mathcal{D}$, is the same as sampling $X \sim \mathcal{D}$.
Applying this observation to our setting, instead of sampling intervals in the sample space of the distribution of interest, we can instead sample intervals uniformly in $[0,1]$, the sample space of $\mathrm{Uniform}(0,1)$, and use the inverse CDF to map these intervals to the sample space of the distribution we want to sample from.
A nice feature of this is that the probability of a uniform random variable being in an interval is equal to the interval's width, so if we want to sample equiprobable intervals, we simply divide the current interval in half and pick a half at random.

The process now looks like this:
\begin{enumerate}
    \item Set the current interval to $[0,1]$.
    \item Divide the current interval in half, and select one of the halves uniformly at random. \label{v2:step:split}
    \item Set the current interval to be the selected half, and denote the endpoints of the current interval as $(a,b)$.
    \item If all values in interval $[F_{\mathcal{D}}^{-1}(a), F_{\mathcal{D}}^{-1}(b)]$ round to the same float $s$, then return $s$.
    \label{v2:step:termination}
    \item Otherwise, repeat from Step \ref{v2:step:split}.
\end{enumerate}
Note that in this updated process, the interval end points $a,b$ are dyadic rationals and therefore can be represented exactly on a computer with finite memory.

There is one more complication to tackle: irrational numbers might still occur in the termination step (\ref{v2:step:termination}).
This brings us to our third and final key idea: approximating the inverse CDF to find an interval $I$ that contains $[F_{\mathcal{D}}^{-1}(a), F_{\mathcal{D}}^{-1}(b)]$, and whose endpoints are rational numbers. 
Then, the termination step becomes:
\begin{enumerate}
  \setcounter{enumi}{3}
  \item Compute an interval $I$ with rational endpoints such that
  % \set{F_{\mathcal{D}}^{-1}(x) \mid x \in [a,b) }
  $[F_{\mathcal{D}}^{-1}(a), F_{\mathcal{D}}^{-1}(b)] \subseteq I$.  If all values in interval $I$ round to the same float $s$, then return $s$.
\end{enumerate}
This approximate interval may be too wide and thus the algorithm may fail to terminate at the appropriate iteration, or fail to terminate entirely (e.g. if the approximation $I$ is always the entire real number line).
Our method of approximating $I$, discussed in~\cref{sec:safe-laplace-mechanism}, has a parameter that controls the precision of the approximation: by increasing the precision in each iteration, we can show that the algorithm eventually terminates with probability 1. 
The fact that the approximation error may cause the algorithm to run for additional iterations is not a problem: once the algorithm has arrived at endpoints $(a,b)$ such that $[F_{\mathcal{D}}^{-1}(a), F_{\mathcal{D}}^{-1}(b)] \subseteq \float^{-1}(s)$ for some $s$, further iterations will not change the outcome, since any subinterval of $[a,b]$ will still map to the same outcome $s$.

\section{Laplace mechanism: algorithm and implementation}
\label{sec:safe-laplace-mechanism}

\def\LD{Lap(\mu, \beta)}
\def\LD{{\mathcal{D}_{\mu, \beta}}}

In this section, we describe our algorithm for safely sampling from a Laplace distribution, explain how it can be the basis for the Laplace mechanism, and report on our experience implementing it.

\cref{alg:lap} is our proposed technique for sampling a Laplace random variable rounded to a float.  Let $\mathcal{D}_{\mu, \beta}$ denote the Laplace distribution with a given location $\mu$ and scale $\beta$.
This algorithm uses a subroutine, \sfunction{IntervalInvCDF}, for computing the inverse CDF of $\LD$ on an interval $[a, b]$.  The subroutine takes as input a tuple $\langle a, b \rangle$ of interval endpoints and a parameter $prec$ that controls the working precision of the function, which influences how closely the returned interval approximates the true interval.  

A little more formally, \sfunction{IntervalInvCDF} satisfies the following properties:

\begin{equation}
F_{\LD}^{-1}(x) \in \sfunction{IntervalInvCDF}_{\mu, \beta} (\langle a, b \rangle, prec) \text{ for } x \in [a, b]    
\end{equation}
and
\begin{equation}
\label{equ:interval-cdf-precision-in-limit}
        \lim_{prec \to \infty} \sfunction{IntervalInvCDF}_{\mu, \beta}(\langle a, b \rangle, prec) = [F_{\LD}^{-1}(a), F_{\LD}^{-1}(b)]
\end{equation}

\cref{alg:lap} starts with the $[0,1]$ interval. On each iteration, it picks one half of this interval uniformly at random and computes the inverse CDF (\cref{alg:lap:intervalInvCDF}) on it to obtain an interval $[s, t]$.  The precision is increased in each iteration to ensure that even  samples near a boundary between two floats can eventually be distinguished.
The algorithm terminates when $s$ and $t$ both round up to the same floating point number. The check for termination happens on \cref{alg:lap:termination} using the \sfunction{NextFloat} function, which takes as input an arbitrary-precision floating point number and rounds it to the next highest 64-bit number.
\cref{fig:laplace} visualizes a run of this algorithm.

\begin{algorithm}
\caption{\textproc{SampleLaplace}$(\mu, \beta)$}
\label{alg:lap}
\begin{algorithmic}[1]
\State $a, b \gets \langle0,1\rangle$
\State $prec\gets 0$
\While{True}
\If {$\sfunction{RandBit}() = 0$}
\State $a \gets \frac{a+b}{2}$
\Else
\State $b \gets \frac{a+b}{2}$
\EndIf
\State $prec\gets prec+1$
\State $\langle s,t\rangle \gets \sfunction{IntervalInvCDF}_{\mu, \beta}(\langle a, b\rangle, prec)$ \Comment{For any $r\in [a, b]$, $F_{\LD}^{-1}(r) \in [s,t] $} \label{alg:lap:intervalInvCDF}
\If {$\sfunction{NextFloat}(s) = \sfunction{NextFloat}(t)$} \label{alg:lap:termination}
\State \Return $\sfunction{NextFloat}(s)$ 
\EndIf
\EndWhile
% \State \textbf{return} $b$
\end{algorithmic}
\end{algorithm}

\def\P{\mathcal{D}_{\text{\sfunction{SampleLaplace}}({\mu,\beta})}}
\def\Pk{\mathcal{D}_{\text{\sfunction{SampleLaplace}}({\mu,\beta,k})}}
\def\Q{\mathcal{D}^{\float}_{\mu,\beta}}

We now formally state the claim that \cref{alg:lap} is arbitrarily close to sampling from a Laplace and rounding to the next float.
Let $\Pk$
be the probability distribution over the outputs of Algorithm~\cref{alg:lap} on inputs $\mu$, $\beta$ after running for at most $k$ iterations; we define the output to be $\bot$ if the algorithm hasn't terminated after $k$ rounds.
 Let $\Q$ be the probability distribution of $\float(X)$ where $X \sim \mathcal{D}_{\mu,\beta}$.

\begin{restatable}{theorem}{main}
\label{thm:algorithm-arbitrarily-close-to-true-distribution}
For any $\mu$, $\beta > 0$ where $\mu$ and $\beta$ are arbitrary-precision floats,
%$$\mathsf{TVD}(\P, \Q) \leq \delta,$$
\begin{equation*}
  \lim_{k \rightarrow \infty} \mathsf{TVD}(\Pk, \Q) = 0
\end{equation*}
where $\mathsf{TVD}$ is the total variation distance.
\end{restatable}

Proof of this theorem appears in Appendix~\ref{sec:proof-main-theorem}.

This algorithm can be the basis for a variant of the Laplace mechanism~\cite{dwork2014the-algorithmic}, which noisily evaluates some target function $f$ on the private data $x$. To do so, it is essential to invoke $\sfunction{SampleLaplace}$ with $\mu=f(x)$ rather than invoking it with $\mu=0$ and adding the result to $f(x)$. The precision-based attacks of~\cref{sec:existing-vulnerabilities} illustrate one of the potential vulnerabilities with the latter approach.

\begin{figure}[t!]
  \centering
  \includegraphics[width=.85\columnwidth]{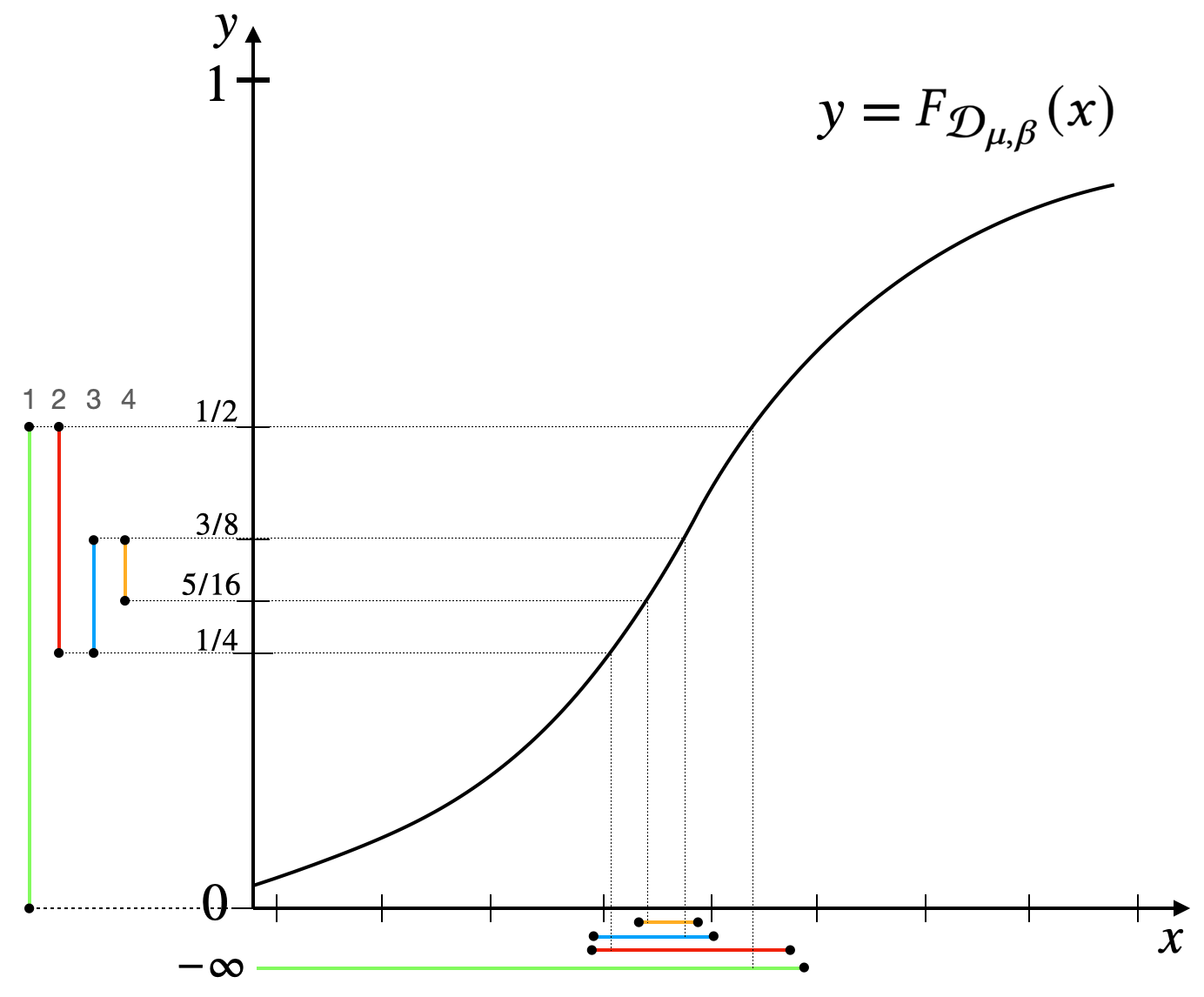}
  \caption{A visualization of \cref{alg:lap}. The green, red, blue and orange lines correspond to the four iterations before the algorithm terminates, and the ticks along the x-axis represent floating point numbers (not drawn to scale). The colored lines below the x-axis represent intervals output by \sfunction{IntervalInvCDF} on the corresponding input intervals $(a, b)$ shown on the y-axis.  Note that each interval contains $\set{F^{-1}_{\LD}(x) \mid a\le x\le b }$, shown by the sub-interval within the dotted lines but may be slightly wider due to the approximation error. 
  It is this approximation that prevents the algorithm from terminating in the third iteration (corresponding to the blue intervals).
  By the fourth iteration (orange intervals), the output interval lies completely between two consecutive floating point numbers which suffices for termination on \cref{alg:lap:termination}.  
  }
  \label{fig:laplace}
\end{figure}

We have implemented \cref{alg:lap} in the Tumult platform~\cite{TumultPlatform}.
Our implementation is in Python.  All steps in the algorithm can be computed \emph{exactly} with the exception of $\sfunction{IntervalInvCDF}$, which is approximate in that it returns an interval that may be \emph{wider} than the ``true'' interval (i.e., $[F_{\mathcal{D}}^{-1}(a), F_{\mathcal{D}}^{-1}(b)]$).  Nevertheless, for the correctness of the algorithm, the returned interval must be guaranteed to contain the true interval.  

To ensure this guarantee, the implementation relies on Arb~\cite{Johansson2017arb}, a C library for arbitrary-precision ball arithmetic.  Ball arithmetic enables computing with real numbers by explicitly and automatically tracking error bounds throughout the computation.  With ball arithmetic, a function $f: \mathcal{X} \rightarrow \mathcal{Y}$ that ordinarily takes a number and returns a number is implemented as a function $F$ that takes in an interval (represented by a midpoint and radius, a ``ball'') and returns an interval with the property that the returned interval contains the correct answer for any number in the input interval (i.e., $\forall x \in [x_1, x_2]: f(x) \in F([x_1, x_2])$).  This correctness guarantee composes nicely and allows us to construct a function like $\sfunction{IntervalInvCDF}$ from the base functions in Arb and be confident the interval it returns contains the true interval.

Arb allows one to increase the working precision which generally speaking reduces the size of the resulting interval, though this is not a universal guarantee and depends on the computation.  For the inverse CDF of the Laplace distribution, which only involves basic arithmetic and a log function, we believe, based on the Arb documentation, that the error decreases exponentially with working precision (cf. Sec. 2.3.6 of ~\cite{arblib-documentation}).

A difference between the implementation in the Tumult platform and \cref{alg:lap} is that instead of sampling a single bit per iteration, we sample 63 bits, thereby choosing among $2^{63}$ equiprobable intervals in each iteration rather than just two.  This improves the runtime of the algorithm by a factor of $\approx$ 40x. 

We simulated ~20 million samples.  The throughput is ~30,000 samples per second.  Almost all samples terminate in a single iteration, but 2\% ($\approx$400,000) took two iterations, and we never witnessed a sample requiring 3 or more iterations.

\section{Extensions, limitations, and future work}
\label{sec:extensions}

While in this paper we describe a technique for safely sampling from a target distribution and rounding to a float, it should be cautioned that this cannot necessarily be used as a drop in replacement in more complex algorithms that rely on sampling from specific distributions as a subroutine.  For some of these algorithms (e.g., the Sparse Vector Technique~\cite{dwork2014the-algorithmic}, PrivTree~\cite{zhang2016privtree}), the proof of correctness leverages properties of the target distribution that may not hold for its rounded variant.

Nevertheless, we believe that our technique, interval refining, can be usefully extended to some of these more complex algorithms. 
To illustrate this point, we briefly describe how it can be used for sampling a noisy argmax, a subroutine that can be used inside the exponential mechanism~\cite{mcsherry2007mechanism} (via the Gumbel-max trick~\cite{GumbelMax}) or as a mechanism itself (e.g., the \sfunction{Report Noisy Max} function~\cite{dwork2014the-algorithmic}).
The noisy argmax problem is as follows: given $n$ private values $x_1, \dots, x_n$, report $\argmax_i x_i + Z_i$ where $Z_i$ is noise sampled independently from some target distribution (e.g., Gumbel).  Note that we cannot simply sample \emph{rounded} noise values as this would not be equivalent.  
But we can still use the interval refining technique: we maintain intervals around the noisy value of each element and we terminate once we have found the largest noisy value (i.e., there is one interval that is strictly larger than all others).   The interval refining techniques gives us flexibility to chooose different termination conditions to suit the particular application.

We are in the process of implementing the interval refining technique for the exponential mechanism~\cite{mcsherry2007mechanism} and an algorithm for noisy quantiles~\cite{smith2011privacy}.  A formal description of the algorithm and a proof of its correctness are left as future work.

A limitation of the interval refinement technique is that the runtime can be difficult to analyze as the number of iterations required depends on the shape of the inverse CDF and the approximation technique used to calculate it.  In practice, with our implementation of the Laplace mechanism, we never witnessed a sample fail to reach the termination condition, with 98\% terminating after a single iteration. A more careful characterization of the runtime is the subject of future work.

\balance

\bibliographystyle{plain}
\bibliography{ref}

\onecolumn

\appendix

\section{Further details on precision-based attacks}
\label{sec:attack-details}

In this section, we provide additional details on the precision-based attacks described in \cref{sec:existing-vulnerabilities}.

We consider double-precision floating-point numbers as specified by IEEE 754~\cite{WikipediaIEEE}, called \emph{doubles} for short.
A double uses 64 bits: 1 sign bit $s$, 11 exponent bits $e$, and 52 mantissa bits $d_1 \dots d_{52}$; the corresponding floating-point number is:
\[
    {\left(-1\right)}^s \cdot {\left(1.d_1 \dots d_{52}\right)} \cdot 2^{e-1023}.
\]
An exponent value of $e=0$ is used to represent 0 and subnormal numbers, while an exponent value of $e=2047$ is used to represent infinities and NaN values; we will ignore these edge cases.
In IEEE 754, arithmetic operations between doubles are \emph{correctly rounded}: they must be computed exactly, and rounded to the closest double.
(In case the result of an operation falls exactly between two successive doubles, special rounding rules apply.)
In the rest of this paper, we denote by $\oplus$, $\ominus$, and $\otimes$ the floating-point analogues of addition, substraction, and multiplication.

The \emph{unit in the last place}, or \emph{ulp}, of a double $x$ is the size of the interval between $x$ and the next consecutive double. We denote it as $\ulp_x$.
For example, the ulp of $x=1$ (represented by $s=0$, $e=1023$, and $d_i=0$ for all $i$) is $2^{-52}$, while the ulp of $x=1-2^{-53}$ (represented by $s=0$, $e=1022$, and $d_i$=1 for all $i$) is $2^{-53}$.
The first fact in \cref{sec:existing-vulnerabilities} follows directly from the definition of doubles: if a double $x$ is such that $2^k \le x < 2^{k+1}$, then $\ulp_x=2^{k-52}$.

Let us now formalize and prove the second observation.

\begin{theorem}\label{thm:sumprecision}
Let $x$ and $y$ be two doubles, with $x \neq 0$.
Then $x \oplus y$ is a multiple of $\ulp_x/2$.
\begin{proof}
Let us assume, without loss of generality, that $x>0$.
Let $k$ be such that of $2^k \le x < 2^{k+1}$; note that $\ulp_x = 2^{k-52}$.
There are two cases, based on whether $y < -2^{k-1}$.
\begin{itemize}
\item If $y < -2^{k-1}$, then $\ulp_y$ is at least $2^{k-1-52}=\ulp_x/2$, so $y$ is a multiple of $\ulp_x/2$.
Since $x$ is a multiple of $\ulp_x$, it's also a multiple of $\ulp_x/2$, and $x \oplus y$ is a multiple of $\ulp_x/2$.
\item If $y \ge -2^{k-1}$, then $x \oplus y \ge - 2^{k-1} + 2^k = 2^{k-1}$.
Then $x \oplus y$ has a ulp of at least $2^{k-1-52}=\ulp_x/2$, so regardless of the value of $y$, the sum will be a multiple of $\ulp_x/2$.
\end{itemize}
\end{proof}
\end{theorem}

Precision-based attacks on additive noise mechanisms follow directly from this fact.

\paragraph*{Additive noise mechanisms}
Distinguishing events can be found simply by adding noise to $0$ or $1$, which can be two outputs of a counting query evaluated on two neighboring datasets.
All noisy values obtained from $1$ are multiples of $2^{-53}$, while a large number of noisy values obtained from $0$ are not.
This affects systems which do not attempt to mitigate floating-point vulnerabilities, like Chorus~\cite{Chorus} or diffpriv~\cite{diffpriv}.
More interestingly, the vulnerability also affects diffprivlib (Gaussian, Analytic Gaussian, and Staircase mechanisms, as well as all variants of the Laplace mechanism except the snapping mechanism), SmartNoise Core (Laplace, Gaussian, and Truncated Gaussian mechanisms), and OpenDP (Laplace, Gaussian, and Analytic Gaussian mechanisms), even though these libraries attempt to mitigate floating-point attacks.

Note that in SmartNoise Core and OpenDP, users need to explicitly opt in to using floating-point primitives using a flag, and are warned that doing so is potentially risky.
Therefore, it could be argued that this problem is technically not a vulnerability.
However, at the time of writing of this paper, some tools relying on these libraries set this flag by default, and do not warn users who are using these primitives; this is the case for e.g. SmartNoise SQL~\cite{SmartNoiseSQL}.

\paragraph*{Quantiles mechanisms}
It is a little less straightforward to find distinguishing events for quantile mechanisms based on the exponential mechanism.
Recall that the naive mechanism first splits the output space in intervals based on the data, then choosing an interval $[x,y)$ using the exponential mechanism, sampling a uniform number $r$ in $[0,1)$, and returning $x \oplus (y \ominus x) \otimes r$.

In diffprivlib, $r$ is generated using the \texttt{SystemRandom.random()} function from Python's standard library, which generates multiples of $2^{-53}$.
This means that with $D_1 = [0, 0, 1]$, all outputs will be multiples of $2^{-53}$.
With $D_2 = [0, 0.25, 1]$ however, whenever the interval $[0, 0.25)$ is selected by the exponential mechanism, and the output of SystemRandom.random() is not a multiple of $2^{-51}$, then the returned number will not be a multiple of $2^{-53}$, creating distinguishing events.

SmartNoise Core attempts to prevent vulnerabilities by using MPFR~\cite{MPFR} to generate $r$ in a hole-free way: all possible doubles in $[0,1)$ can be generated.
As a consequence, the previous choice of $D_1$ and $D_2$ does not create distinguishing events.
However, if we use $D_1 = [-1, 1, 1]$, then the only possible sampled interval is $[-1, 1)$: because of the addition with $1$, all outputs will be multiples of $2^{-53}$.
Using $D_2 = [-1, 0, 1]$, the interval $[0, 1)$ might be sampled, in which case the output might not be a multiple of $2^{-53}$.

Interestingly, SmartNoise Core does \emph{not} block the use of this quantiles mechanism on the user-specified floating-point option which is necessary to access floating-point additive noise mechanisms.
This underscores that floating-point vulnerabilities can occur in places where they would not be expected, and thereby evade scrutiny.

\section{Proof of Theorem~\ref{thm:algorithm-arbitrarily-close-to-true-distribution}}
\label{sec:proof-main-theorem}

We begin by recalling the statement of the theorem:

\main*

To simplify notation we let $P^{(k)} = \Pk$ and $Q = \Q$.
Recall that $S$ is the set of rational numbers representable as floats and that $P^{(k)}$ is a distribution over $S \cup \{\bot\}$\footnote{While $\bot$ is not a possible outcome for distribution $Q$, we can consider it as a possible outcome with probability 0, and thus the TVD is well-defined.}, where $s\in S$ is the event that the algorithm terminates in the first $k$ rounds and outputs $s$, and $\bot$ is the event that the algorithm does not terminate after $k$ rounds.
$Q$ is the distribution over $S$ that results from picking a real number from $\mathcal{D}_{\mu,\beta}$ and rounding it to the nearest float.

The proof the theorem has the following structure:
\begin{enumerate}
  \item We show that for all $k \in [1,\infty)$ and for all $s \in S$, $P^{(k)}(s) \leq Q(s)$ (Claim~\ref{clm:probability-lower-bound}).
  \item We show that $\lim_{k \rightarrow \infty} P^{(k)}(\bot) = 0$ (Claim~\ref{clm:probability-non-termination}). That is, the probability that the algorithm does not terminate within $k$ rounds goes to zero as $k$ goes to infinity.
  \item We show that (1) and (2) together imply that the theorem holds.
\end{enumerate}

Before we can prove the first of these claims, we need to show some properties about the intervals produced in the intermediate stages of the algorithm.
We call the interval that the algorithm sets to $[a_{k},b_{k}]$ at the beginning of each iteration the \emph{initial interval} for that iteration.
We let $P_{init}^{(k)}$ denote the probability distribution over choices of initial interval that result in the algorithm terminating in that round.
That is, $P^{(k)}_{init}(I)$ is the probability that the algorithm picks $I$ as its initial interval in some round, and then terminates at the end of the round.
Note that whether the algorithm terminates on a given initial interval is deterministic, so the algorithm always terminates in the first round that it chooses an initial interval in $\supp\left( P^{(k)}_{init} \right)$.
We let $P^{(k)}_{init}(\bot)$ be the probability that the algorithm does not terminate in the first $k$ rounds.
See Figure~\ref{fig:initial-interval-example} for an illustration of the possible initial intervals chosen by the Algorithm.

\begin{figure}[h]
  \centering
  \includegraphics[width=\textwidth]{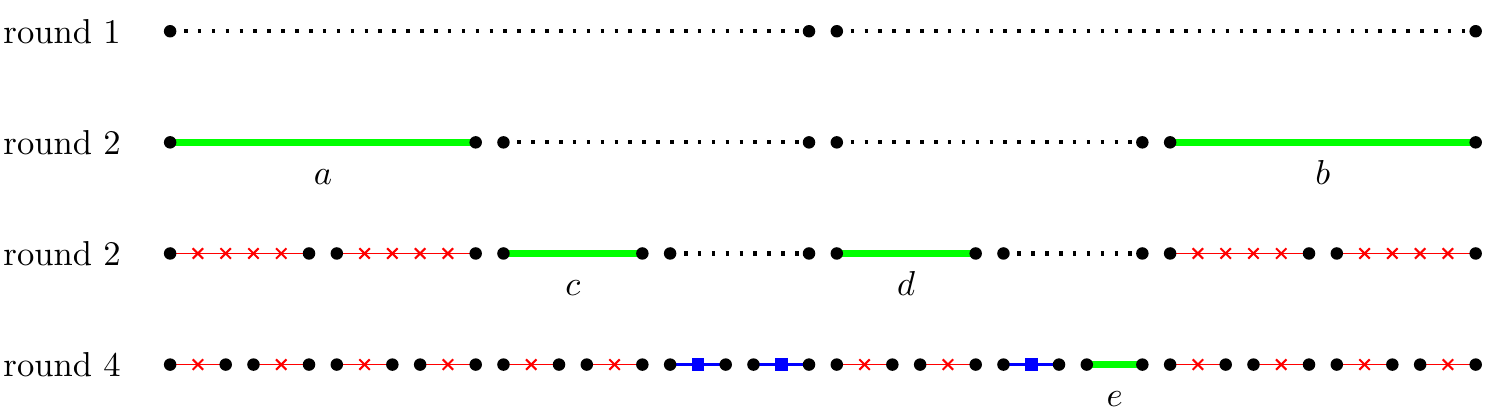}
  \caption{An illustration of the possible initial intervals chosen by the 4-round version of the algorithm. In each round, the algorithm picks one of the two sub intervals of the interval it chose in the previous round. The solid green intervals represent intervals that, if chosen, cause the algorithm to terminate in the given round. Dotted gray intervals represent intervals that could be chosen as the initial interval in a given round, but will never cause the algorithm to terminate. Red intervals with x's represent intervals that can never be chosen by the algorithm because they are sub intervals of the green intervals. Finally, blue intervals marked with squares represent intervals in the last round that may be chosen and cause the algorithm to output $\bot$. The distribution $P^{(4)}_{init}$ is a distribution over intervals $a$ through $e$ and $\bot$, with probabilities $1/4, 1/4, 1/8, 1/8, 1/16$ and $3/16$ respectively. Note that the green intervals are disjoint (see Claim~\ref{clm:algorithm-intervals-disjoint} for a proof).}
  \label{fig:initial-interval-example}
\end{figure}

Next, we want to consider the distributions that result when this initial interval is transformed by the inverse CDF function.
We let $F^{-1}$ denote the interval version of the inverse CDF function for distribution $\mathcal{D}_{\mu,\beta}$.
That is $F^{-1}(I)$ is inverse CDF applied to each endpoint of interval $I$.
Additionally, we let $\tilde{F}^{-1}_{k}$ denote the imprecise interval inverse CDF function with precision $k$, $\sfunction{IntervalInvCDF}_{\mu, \beta}(\cdot, k)$.
We consider two distributions: the one resulting from applying $F^{-1}$ to $P^{(k)}_{init}$, denoted $P^{(k)}_{F^{-1}}$, and the one resulting from applying $\tilde{F}^{-1}_{k}$ to $P^{(k)}_{init}$, denoted $P^{(k)}_{\tilde{F}^{-1}}$ (we drop the subscript $k$ on $\tilde{F}^{-1}$ since it is already present in the superscript).
The probability of $\bot$ in each distribution is the same, i.e. $P^{(k)}_{init}(\bot) = P^{(k)}_{F^{-1}}(\bot) = P^{(k)}_{\tilde{F}^{-1}}(\bot)$.

The following two claims (Claim~\ref{clm:algorithm-intervals-disjoint} and \ref{clm:algorithm-intervals-correct-probability}) help to prove Claim~\ref{clm:probability-lower-bound}.
We first show that intervals in the support of $P^{(k)}_{init}$ are disjoint, and so are the intervals in the support of $P^{(k)}_{F^{-1}}$ (however, the same is not necessarily true of $P^{(k)}_{\tilde{F}^{-1}}$).

\begin{claim}
  \label{clm:algorithm-intervals-disjoint}
  All intervals in $\supp(P^{(k)}_{init})$ are pairwise disjoint.
  Additionally, all intervals in $\supp (P^{(k)}_{F^{-1}})$ are pairwise disjoint.
\end{claim}
\begin{proof}

  Suppose a pair of intervals $I, I' \in \supp \left(P^{(k)}_{int}\right)$ is overlapping.
  Note that each interval only has the possibility of being chosen by the algorithm on some fixed round (since the length of the intervals the algorithm considers halves every round).
  We let $round(I)$ denote this round for interval $I$.

  Next, we claim that if the algorithm chooses interval $I$ as its initial interval in round $round(I)$, then it must be the case that the algorithm terminates in that round.
  Since the interval is in the support of the distribution, the algorithm must terminate with some positive probability after choosing $I$.
  However, whether the algorithm terminates after choosing an interval in a given round is deterministic, and therefore the algorithm either always terminates for a given interval, or never terminates.

  Finally, note that any two intervals chosen by the algorithm can only intersect if one is strictly contained within the other.
  Assume WLOG that $I' \subseteq I$, and therefore $round(I) < round(I')$.
  On any run of the algorithm that terminates in round $round(I')$ after choosing $I'$, it must be the case that on round $round(I)$, the algorithm chose $I$ as its initial interval.
  Therefore, the algorithm should have terminated in $round(I)$, giving a contradiction.

 The same property holds for $P^{(k)}_{F^{-1}}$ since the inverse CDF function is monotonically increasing, and therefore a pair of intervals in $\supp \left(P^{(k)}_{F^{-1}}\right)$ that the corresponding pair of intervals in $\supp\left(P^{(k)}_{init}\right)$ would overlap.
\end{proof}

\begin{claim}
  \label{clm:algorithm-intervals-correct-probability}
  For all $I \in \supp \left(P^{(k)}_{F^{-1}}\right)$, $I \ne \bot$,
  \begin{equation}
    P^{(k)}_{F^{-1}}(I) = \mathcal{D}_{\mu,\beta}(I).
  \end{equation}
\end{claim}

\begin{proof}
  First, we claim that for all $J \in \supp\left(P^{(k)}_{init}\right)$, $P^{(k)}_{init}(J) = |J|$.
  Note that for the algorithm to choose interval $J$, it must have chosen the only interval in each of the previous rounds that contains $J$, i.e. there is only one sequence of choices that leads to $J$ being chosen.
  Let $round(J)$ be defined as in the proof of Claim~\ref{clm:algorithm-intervals-disjoint}.
  Let $y = round(J)$, and let $J_{1}, \ldots, J_{y-1}$ be the sequence of initial intervals the algorithm chose in the first $y-1$ rounds.
  Then, the algorithm will never terminate after choosing any $J_{i}$ (by Claim~\ref{clm:algorithm-intervals-disjoint}), and the probability that the algorithm chooses $J_{i}$ given that it chose $J_{i-1}$ is $1/2$.
  The claim follows from induction over $J_{i}$.

  Next, fix $I \ne \bot \in \supp\left(P^{(k)}_{F^{-1}}\right)$.
  Let $J = F(I)$ be the corresponding initial interval that the algorithm chooses.
  Then,
  \begin{align}
    P^{(k)}_{F^{-1}}(I) &= P^{(k)}_{init}(J) \nonumber \\
    &= |J| \nonumber \\
    &= F(I[1]) - F(I[0]) \label{equ:non-interval-cdf}\\
    &= \mathcal{D}_{\mu,\beta}(I). \nonumber
  \end{align}
  The notation $F$ in equation \ref{equ:non-interval-cdf} denotes the non-interval CDF function.
\end{proof}

\begin{claim}
  \label{clm:probability-lower-bound}
  For all $k \in [1,\infty)$ and for all $s \in S$, $P^{(k)}(s) \leq Q(s)$.
\end{claim}
\begin{proof}
  Let $\mathcal{I}_{s}$ be the subset of $\supp\left(P^{(k)}_{F^{-1}}\right)$ that results in the algorithm producing output $s$.
  That is, if we let $float^{-1}(s)$ be the preimage of float $s$ under function $float$, then $\mathcal{I}_{s} = \left\{I \in \supp\left(P^{(k)}_{F^{-1}}\right) \text{ s.t. } \tilde{F}^{-1}(F(I)) \subseteq float^{-1}(s)\right\}$.
  Then,
  \begin{align*}
    P^{(k)}(s) &= P^{(k)}_{F^{-1}} (\mathcal{I}_{s}) \tag{by definition of $\mathcal{I}_{s}$} \\
    &= \sum_{I \in \mathcal{I}_{s}} P^{(k)}_{F^{-1}} (I) \\
    &= \sum_{I \in \mathcal{I}_{s}} \mathcal{D}_{\mu,\beta}(I) \tag{by Claim~\ref{clm:algorithm-intervals-correct-probability}} \\
    &= \mathcal{D}_{\mu,\beta}\left( \bigcup_{I \in \mathcal{I}_{s}} I \right) \tag{by Claim~\ref{clm:algorithm-intervals-disjoint}} \\
    &\le Q(s). \tag{by definition of $\mathcal{I}_{s}$}
  \end{align*}
\end{proof}

\begin{claim}
  \label{clm:probability-non-termination}
  $\lim_{k \rightarrow \infty} P^{(k)}(\bot) = 0$. That is, the probability that the algorithm does not terminate within $k$ rounds goes to zero as $k$ goes to infinity.
\end{claim}
\begin{proof}
  We bound the probability that the algorithm does not terminate in $k$ rounds, given that the algorithm reached round $k$.
  This is an upper bound on the probability that the algorithm does not terminate in $k$ rounds, $P^{(k)}(\{\bot\})$.

  Let $I$ denote the algorithm's choice of an initial interval for round $k$.
  The algorithm does not terminate if and only if $s \in \tilde{F}^{-1}_{k}(I)$ for some $s \in S$.
  We consider two cases:

  \textbf{Case 1:} $s \in F^{-1}(I)$ for some $s \in S$.
  There are only $|S|$ intervals in $\supp\left(P^{(k)}_{F^{-1}}\right)$ for which this is the case (by Claim~\ref{clm:algorithm-intervals-disjoint}), and for all such intervals $J$, $P^{(k)}_{F^{-1}}(J) = 2^{-k}$.
  Therefore, letting $\mathcal{I}_{1}$ denote this set of intervals,
  \begin{equation}
    P^{(k)}_{F^{-1}}(\mathcal{I}_{1}) \le 2^{-k} |S|.
  \end{equation}

  \textbf{Case 2:} $s \not\in F^{-1}(I)$ for any $s \in S$ (but $\tilde{F}^{-1}_{k}$ does contain some $s \in S$). Then, we use our assumption about the precision of $\tilde{F}^{-1}_{k}$ (Equation~\ref{equ:interval-cdf-precision-in-limit}).
  In particular, Equation~\ref{equ:interval-cdf-precision-in-limit} implies that there is a function $\delta(k)$ such that
  \begin{equation*}
    \lim_{k \rightarrow \infty} \delta(k) = 0,
  \end{equation*}
  and for any interval $J$,
  \begin{equation}
    d\left(J, \tilde{F}^{-1}_{k}(I)\right) \le \delta(k),
  \end{equation}
  where $d(I,J)$ is defined as $|I[0] - J[0]| + |I[1] - J[1]|$.
  In particular, this true of $J = F^{-1}(I)$.
  Then, the condition for case 2 is only possible if at least one of the endpoints of $F^{-1}(I)$ is within distance $\delta(k)$ of some $s \in S$.
  We split this into two subcases.

  \textbf{Case 2(a):}
  Only one of the endpoints of $F^{-1}(I)$ is within distance $\delta(k)$ of some $s \in S$ (the other endpoint may be far from all elements of $S$, or it may be close to some other $t \in S, t \not = s$).
  By Claim~\ref{clm:algorithm-intervals-disjoint}, this can be true of at most two intervals in $\supp\left(P^{(k)}_{F^{-1}}\right)$ per element of $S$.
  Denote this set of intervals $\mathcal{I}_{2}$.
  There are at most $2|S|$ of these intervals, and each has probability $2^{-k}$ of being chosen, so
  \begin{equation}
    P^{(k)}_{F^{-1}}(\mathcal{I}_{2}) \le 2^{1-k} |S|.
  \end{equation}

  \textbf{Case 2(b):}
  There exists an $s \in S$ such that both endpoints of $F^{-1}(I)$ are within $\delta(k)$ of $s$.
  Let $\mathcal{I}_{3}$ denote the set of the intervals of $\supp\left(P^{(k)}_{F^{-1}}\right)$ meeting this condition.
  Then, since the intervals in $\mathcal{I}_{3}$ are disjoint, the total length of the intervals in $\mathcal{I}_{3}$ is at most $2 \cdot \delta(k) \cdot |S|$.
  So,
  \begin{align*}
    P^{(k)}_{F^{-1}}(\mathcal{I}_{3}) &= \mathcal{D}_{\mu,\beta}\left(\bigcup_{J \in \mathcal{I}_{3}} J\right) \\
    &\le \Omega_{\mathcal{D}_{\mu,\beta}}(2 \cdot \delta(k) \cdot |S|),
  \end{align*}
  where $\Omega_{\mathcal{D}}(x)$ denotes the maximum probability that distribution $\mathcal{D}$ assigns to a subset of reals of total length at most $x$.
  For the Laplace distribution, $lim_{|x| \rightarrow 0} \Omega_{\mathcal{D}}(x) = 0$.

  Combining these cases we get
  \begin{align*}
    P^{(k)}(\bot) &\le P^{(k)}_{F^{-1}}(\mathcal{I}_{1}) + P^{(k)}_{F^{-1}}(\mathcal{I}_{2}) + P^{(k)}_{F^{-1}}(\mathcal{I}_{3}) \\
    &\le 2^{-k}|S| + 2^{1-k} |S| + \Omega_{\mathcal{D}_{\mu,\beta}}(2 \cdot \delta(k) \cdot |S|) \\
    &\defeq \alpha(k).
  \end{align*}
  Then the claim follows since $\lim_{k \rightarrow \infty} \alpha(k) = 0$.
\end{proof}

Finally, we show how Claims~\ref{clm:probability-lower-bound} and \ref{clm:probability-non-termination} imply that $\lim_{k \rightarrow \infty} TVD( P^{(k)}, Q) = 0$.
\begin{align*}
  \mathsf{TVD}(P^{(k)}, Q) &= \frac{1}{2} \sum_{s \in S \cup \{\bot\}} \left|P^{(k)}(s) - Q(s)\right| \\
  &= \frac{1}{2} \sum_{s \in S} \left|P^{(k)}(s) - Q(s)\right|  + \frac{1}{2} P^{(k)}(\bot) \tag{since $\bot$ is not in the support of Q} \\
  &= \frac{1}{2} \sum_{s \in S} \left( Q(s) - P^{(k)}(s) \right) + \frac{1}{2} P^{(k)}(\bot) \tag{by Claim~\ref{clm:probability-lower-bound}} \\
  &= \frac{1}{2} - \frac{1}{2} \sum_{s \in S} P^{(k)}(s) + \frac{1}{2} P^{(k)}(\bot) \\
  &= \frac{1}{2} - \frac{1}{2} P^{(k)}(S) + \frac{1}{2} P^{(k)}(\bot) \\
  &= \frac{1}{2} - \frac{1}{2} \left( 1 - P^{(k)}(\bot)  \right) + \frac{1}{2} P^{(k)}(\bot) \\
  &= P^{(k)}(\bot)
\end{align*}

Therefore, by Claim~\ref{clm:probability-non-termination},
\begin{equation}
  \lim_{k \rightarrow \infty} \mathsf{TVD}(P^{(k)}, Q) = 0
\end{equation}
which completes the proof.
\end{document}